\documentclass{amsart}

\usepackage{amsmath,amsthm,amssymb}
\usepackage{graphicx}
\usepackage{tikz-cd}
\usepackage{paralist}
\usepackage{environ}
\usepackage{xcolor}
\usepackage{hyperref}
\usepackage{centernot}
\usepackage{mathtools}
\usepackage{marvosym}
\usepackage{xcolor,cancel}

\usepackage{verbatim}

\usepackage{pstricks}

\usepackage{mdframed}

\usepackage{subfigure}
\renewcommand{\thesubfigure}{\thefigure.\arabic{subfigure}}
\makeatletter
\renewcommand{\p@subfigure}{}
\renewcommand{\@thesubfigure}{\thesubfigure:\hskip\subfiglabelskip}
\makeatother

\makeatletter
\def\square{\pst@object{square}}
\def\square@i(#1,#2)#3{{\use@par\solid@star\psframe[origin={#1,#2}](#3,#3)}}
\makeatother

\makeatletter
\DeclareFontFamily{U}{tipa}{}
\DeclareFontShape{U}{tipa}{bx}{n}{<->tipabx10}{}
\newcommand{\arc@char}{{\usefont{U}{tipa}{bx}{n}\symbol{62}}}%

\newcommand{\arc}[1]{\mathpalette\arc@arc{#1}}

\newcommand{\arc@arc}[2]{%
  \sbox0{$\m@th#1#2$}%
  \vbox{
    \hbox{\resizebox{\wd0}{\height}{\arc@char}}
    \nointerlineskip
    \box0
  }%
}
\makeatother

\makeatletter
\newcommand{\doublewedge}{\big@doubleop{\wedge}}
\newcommand{\big@doubleop}[1]{%
  \DOTSB\mathop{\mathpalette\big@doubleop@aux{#1}}\slimits@
}

\newcommand\big@doubleop@aux[2]{%
  \sbox\z@{$\m@th#1#2$}%
  \makebox[1.35\wd\z@][s]{$\m@th#1#2\hss#2$}%
}

\newcommand{\abs}[1]{\left|#1\right|}     

\newcommand{\dnear}{\delta_{\Phi}} 
\newcommand{\rnear}{\widetilde{\delta}_{\Phi}} 

\theoremstyle{plain}
\newtheorem{axiom}{Axiom}
\newtheorem{theorem}{Theorem}
\newtheorem{lemma}{Lemma}
\newtheorem{remark}{Remark}
\newtheorem{definition}{Definition}

\newtheorem{example}{Example}

\makeatletter
\@namedef{subjclassname@2020}{\textup{2020} Mathematics Subject Classification}
\makeatother

\usepackage{scalerel}
\usepackage{stackengine,amsmath}

\begin{document}

\title{Near Stein-Weiss Finite Vector Field Groups in\\ Characteristic Nearness Approximation Spaces\\
in the polar Complex Plane}

\author[J.F. Peters]{James F. Peters}
\address{
Department of Electrical and Computer Engineering,
University of Manitoba, WPG, Manitoba, R3T 5V6, Canada and
Department of Mathematics, Faculty of Arts and Sciences, Ad\.{i}yaman University, 02040 Ad\.{i}yaman, Turkey
}
\email{james.peters3@umanitoba.ca}

\author[M.A. \"{O}t\"{u}rk]{Mehmet A. \"{O}zt\"{u}rk}
\address{
Department of Mathematics, Faculty of Arts and Sciences, Ad\.{i}yaman University, 02040 Ad\.{i}yaman, Turkey
}
\email{mehaliozturk@gmail.com}


\date{}
  
\begin{abstract}
This paper introduces results for characteristically near Stein-Weiss groups inherent in vector fields in the complex plane $\mathbb{C}$.  Near groups are discerned in the context of characteristic nearness approximation spaces (cNASs). A characteristic of a Stein-Weiss group is a holomorphic mapping $\varphi:\Omega\in 2^{\mathbb{C}}\to\mathbb{C}$ defined by $\varphi(t)=e^{jt}$, which defines a vector field in the complex plane.   All characteristic vectors emanate from the same fixed point in $\mathbb{C}$, namely, 0.  
\end{abstract}


\maketitle
\tableofcontents

\section{Introduction}
Characteristically near groups arise from variants of the characteristic function $\varphi(t)= e^{jt}$ from~\cite[\S IV.5, p. 173]{Stein1971}, which defines the characteristic of a Stein-Weiss multiplicative abelian group $(\left\{\varphi(z_0\in\mathbb{C})\right\},\cdot)$, i.e., 
\begin{center} 
\boxed{\boldsymbol{
\varphi:\Omega\in 2^{\mathbb{C}}\to\mathbb{C}\ \mbox{defined by}\ \varphi(z_0) = e^{2\pi iz_0}: (G,\cdot)=(\left\{\varphi(z_0)\right\},\cdot)
,t\in \mathbb{R},
}}
\end{center}
Characteristic functions were introduced by W.R. Hamilton in 1837 in a study of light rays~\cite{Hamilton1837} and tailored for various other settings by J.L. Synge in 1931~\cite{Synge1931}.  
The Stein-Weiss group $(G,\cdot)$ (inspired L. Euler by~\cite[\S VIII, p. 112]{Euler1748}) is a natural biproduct 
of a vector field in the complex plane that reflects the behaviour of a dynamical system.  The characteristic vector \boxed{\vec{\Phi}(G) = (\varphi_1(G),\varphi_2(G),\dots,\varphi_n(G))} contains characteristics $\varphi_i(t\in\Omega)\in [-1,1]\in\mathbb{C}$ of $\vec{\Phi}(G)$ that reflect the structure $G$, e.g.,

\vspace*{0.2cm}

\begin{example}{\bf (Sample Group Characteristics).}\\
\begin{align*}
\varphi_1(G) &=k\ \in \left\{0\mbox{(not finite)},1\mbox{(finite)}\right\}.\\
\varphi_2(G) &=k\ \in \left\{0\mbox{(not abelian)},1\mbox{(abelian)}\right\}.\\
\varphi_3(G) &= o(G)\in \mathbb{Z}^{0+}\ \mbox{(order of $G$)}.\\
\varphi_4(G) &=\ \abs{\vec{G}}\in 2^{\mathbb{C}}\ \mbox{(size of finite vector field in complex plane)}.\\
\varphi_5(G) &= \ \abs{\lambda_{max}}\in \left\{\lambda_{\vec{G}}\right\}\ \mbox{maximum eigenvalue in}\ \left\{\lambda_{\vec{G}}\right\}.\mbox{\qquad \textcolor{blue}{$\blacksquare$}}. 
\end{align*}
\end{example}

\vspace*{0.2cm}
  
Near groups are discerned in the context of a characteristic nearness approximation space (cNAS), which is introduced in this paper.  Unlike the earlier
nearness approximation space (NAS~\cite{Peters1}) and weak nearness approximation space (wNAS) introduced in~\cite{wNAS2019}, the backbone of a cNAS is characteristic function $\varphi(t)$ accompanied by a characteristic distance $d^{\boldsymbol{\Phi}}(A,B)$ measurement between object $A$ and $B$ (see Def.~\ref{def:dnear0}), introduced in~\cite{peters2025}.
It is the characteristic distance that provides a formal basis for characteristic equivalences and partition of an object space into equivalence classes (see Def.~\ref{def:cNAS}).

\begin{figure}[!ht]
	\centering
\includegraphics[width=74mm]{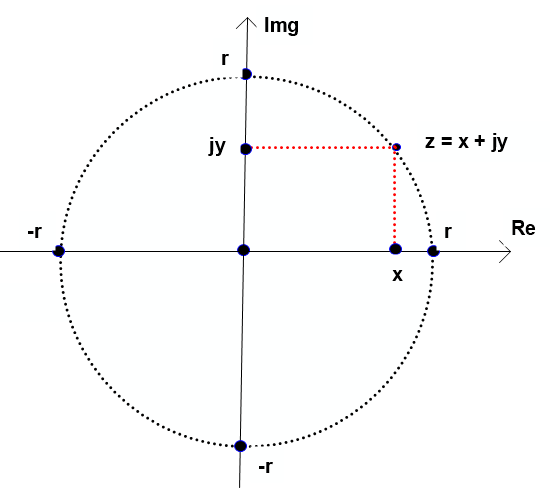}
\caption{Complex Plane.}
	\label{fig:C}
\end{figure}
 
\begin{table}[h!]\label{table:symbols}
\begin{center}
\caption{Principal Symbols Used in this Paper}
\begin{tabular}{|c|c|}
\hline
Symbol & Meaning\\ 
\hline\hline
$\boldsymbol{\Box}$ & End of Proof.\\
\hline
$\mbox{\textcolor{blue}{$\blacksquare$}} $ & End of block.\\
\hline
$\mathbb{C}$ & Complex plane.\\
\hline
$j$ & $j^2 = -1,\ \mbox{also written}$\ $i: i^2 = -1$.\\
\hline
$\vec{0}$ & center of unit circle in polar $\mathbb{C}$.\\
\hline 
$z$ & = $a + jb=e^{j\theta}\in\mathbb{C}:a,b\in\mathbb{R}$ (see Fig.~\ref{fig:C}).\\
\hline
$2^X$ & $\mbox{collection of subsets in $X$}$.\\
\hline
$A\ \dnear\ B$ & $A$ characteristically near $B$.\\
\hline
$\varphi(a\in A)\in \mathbb{C}$ & Characteristic of $a\in A$.\\
\hline
$\Phi(A)$ & =$\left\{\varphi(a_i):i\geq 1,a_i\in A \right\}\in2^{\mathbb{C}}$.\\
\hline
$d^{\Phi}(A,B)$ & Infimum of Characteristic Distances between $A$ and $B$.\\
\hline
cNAS &  characteristic Nearness Approximation Space.\\
\hline
\end{tabular}
\end{center}
\end{table}


\section{Preliminaries}
\noindent Detected affinities between vector fields for stable systems result from determining the infimum of the characteristic distances between pairs of system characteristics.\\

\vspace*{0.2cm}  

\begin{definition}\label{def:vector}{\bf (Vector).}\\
A {\bf vector} $v$ (denoted by $\vec{v}$) is a quantity that has magnitude and direction in the complex plane $\mathbb{C}$.
\qquad \textcolor{blue}{$\blacksquare$}
\end{definition}

\begin{figure}[!ht]
	\centering
\includegraphics[width=125mm]{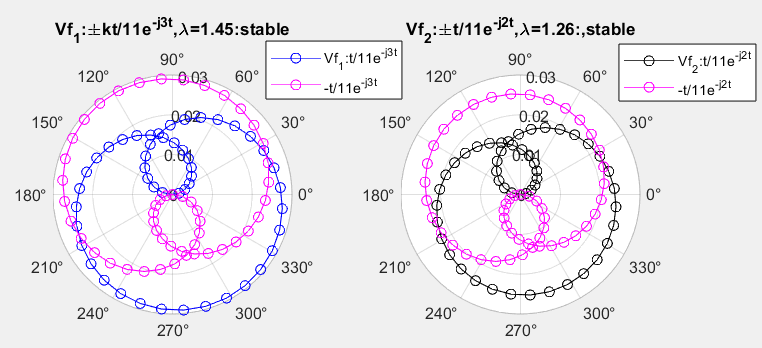}
  \caption{Characteristically Near Vector Fields in Polar Complex Plane:
	\boxed{\mbox{[left]$\vec{V}f_1$},stable},\boxed{\mbox{[right]$\vec{V}f_2$},stable}}
	\label{fig:nearVfs}
\end{figure} 

\begin{definition}\label{def:Vf0}{\bf (Vector Field in the Complex Plane).}\\
Let $U = \left\{p\in \mathbb{C}\right\}$ be a bounded region in the complex plane containing points $p(x,jy)\in U$. A {\bf vector field} is a mapping $F:U\to 2^{\mathbb{C}}$ defined by\\ 

\vspace*{0.2cm}

\begin{center}
\boxed{\boldsymbol{
F(p(x,jy)) = \left\{\vec{v}\right\}\in 2^{\mathbb{C}}\ 
\mbox{\bf (denoted by\ $\vec{V}f$)}
\mbox{\qquad \textcolor{blue}{$\blacksquare$}}
}}
\end{center}
\end{definition}

\vspace*{0.2cm}

\begin{mdframed}[backgroundcolor=green!15]
\begin{axiom}\label{axiom: PhysSysRepresentation}{\bf (Physical System Representation Axiom).}\\
Every physical system has a represention as a vector field in the complex plane.
\qquad \textcolor{blue}{$\blacksquare$}
\end{axiom}
\end{mdframed}\mbox{}\\

\vspace*{0.2cm}

\begin{remark}{\bf (Practical Outcomes of Axiom~\ref{axiom: PhysSysRepresentation}).}\\
Axiom~\ref{axiom: PhysSysRepresentation} has a number of practical outcomes in the design and evaluation of systems.  Let $\vec{V}f(\mathcal{O})$ be a vector field representation of a system $\mathcal{O}$ (collection of objects in a system).  Then
\begin{compactenum}
\item Every object $x\in \mathcal{O}$ is represented by a vector 
\boxed{\boldsymbol{\vec{x}\in\vec{V}f(\mathcal{O})\in \mathcal{C}}}.
\item Every object $x\in \mathcal{O}$ is measurable in terms of a radius $r_{\vec{V}f(\mathcal{O})}$:\\

\vspace*{0.2cm}
 
\begin{center}
\boxed{\boldsymbol{
r_{\vec{V}f(\mathcal{O})} =\ \mbox{distance of the tip of}\ \vec{v}\ \ \mbox{from the origin}\ \vec{0}\in \vec{V}f(\mathcal{O}).
}}
\end{center}

\vspace*{0.5cm}

and angle $\theta_{\vec{v}f(\mathcal{O})} $:\\

\vspace*{0.2cm}

\begin{center}
\boxed{\boldsymbol{
\theta_{\vec{V}f(\mathcal{O})} =\ \mbox{direction of}\ \vec{v}\in \vec{V}f(\mathcal{O}).
}}\ \mbox{\qquad \textcolor{blue}{$\blacksquare$}}
\end{center}

\vspace*{0.2cm}

\end{compactenum} 
\end{remark}

\vspace*{0.2cm}

\begin{remark}{\bf (Polar Form of a Complex Number).}\\ 
Recall that a complex number $z$ in polar form (introduced by Euler~\cite[\S VIII, p. 112]{Euler1748}) is written \boxed{\boldsymbol{z=re^{j\theta}=cos\theta+jsin\theta}}. For the geometry of the polar complex plane plots of $re^{j\theta}$ like those in Fig.~\ref{fig:polarZ} and Fig.~\ref{fig:nearVfs}, see T. Needham~\cite{Needham1997}.
\qquad \textcolor{blue}{$\blacksquare$}
\end{remark}

\vspace*{0.2cm}

\begin{definition}\label{def:Krantz}{\bf (Krantz Vector Field Stability Condition~\cite{Krantz2010}).}\\
A vector field $\vec{V}f$ in the complex plane is stable, provided each $\vec{V}f$ lies either within or on the boundary of the unit circle centered at $\boldsymbol{0}$ in $\mathbb{C}$. \qquad \textcolor{blue}{$\blacksquare$}
\end{definition}

\vspace*{0.2cm}
\begin{example}
Two examples of stable vector fields in polar form are given in Fig.~\ref{fig:nearVfs}. In both cases, observe that all of the vectors in the fields in Fig.~\ref{fig:nearVfs} lie within the boundary of the unit circle in the complex plane.  Hence, from Def.~\ref{def:Krantz}, both vector fields are stable. 
\qquad \textcolor{blue}{$\blacksquare$}
\end{example}

\vspace*{0.2cm}


\begin{example}\label{ex:2vectorFields}{\bf (Pair of Vector Fields).}\\
Two sample vector fields \boxed{\vec{V}f_1,\vec{V}f_2} are shown in Fig.~\ref{fig:nearVfs}, namely,
\begin{align*}
\vec{V}f_1:\mathbb{R} &\to \mathbb{C},\ \mbox{defined by}\\
   & \vec{V}f_1(t) = \pm\frac{kt}{11}e^{-j3t}\ \mbox{\bf (lefthand Vector field in Fig.~\ref{fig:nearVfs})}.\\
\vec{V}f_2:\mathbb{R} &\to \mathbb{C},\ \mbox{defined by}\\
   & \vec{V}f_2(t) = \pm\frac{t}{11}e^{-j2t}\ \mbox{\bf (righthand Vector field in Fig.~\ref{fig:nearVfs})\qquad \textcolor{blue}{$\blacksquare$}}.\\ 					
\end{align*}
\end{example}

\begin{lemma}\label{lemma:charGroup}{\bf (Stein-Weiss Multiplicative Vector Field Group~\cite[\S IV.5, p. 173]{Stein1971}).}\\
The vector field $\vec{V}f(t) = e^{\pm jt}$ with multiplication $\cdot$ on its exponentials is an abelian Stein-Weiss group
\boxed{G(\left\{e^{\pm jt}\right\},\cdot)}.
\end{lemma}
\begin{proof}
Let $\vec{V}f = \left\{e^{\pm jt}\in\mathbb{C}\right\}$ be a vector field in the complex plane and observe
\begin{align*}
e^0 &= \vec{1}\in \vec{V}\ \mbox{\bf (multiplicative identity), since}\\
    & e^{\pm jt}e^0 = e^{\pm jt + 0} = e^{\pm jt}.\\
\vec{z}\in \vec{V}f  & = e^{jt}\\
                     &= cos t + jsin t\in \mathbb{C}.\\
\vec{z}^* &= e^{-jt}\ \mbox{\bf (conjugate of $\vec{z}$).}\\
          &= cos t - jsin t\in \mathbb{C}\ \mbox{\bf (multiplicative inverse of $e^{jt}$)}, i.e.,\\								
\vec{z}\cdot \vec{z}^* &= e^{jt}e^{-jt} = e^{j(t-t)} = e^0.  
\end{align*}
Hence, \boxed{G(\left\{e^{\pm jt}\right\},\cdot)} is a vector field group.  From 
\begin{center}
\boxed{\boldsymbol{
e^{jt}e^{jt'} = e^{jt+jt'} = e^{j(t+t')} = e^{j(t'+t)}
= e^{jt'}e^{jt},
}}
\end{center}
the Stein-Weiss group is abelian.
\end{proof}

\begin{theorem}\label{lemma:charGroup}{\bf (Stein-Weiss Additive Vector Field Group).}\\
Let the vector field $\vec{V}f:\mathbb{C}\to\mathbb{C}$ be defined by
$\vec{V}f(\pm z) = \pm z$. 
\boxed{G(\left\{\vec{V}f(\pm z)\right\},+)} is an additive abelian vector field group. 
\end{theorem}
\begin{proof}
Let $\vec{V}f = \left\{\vec{V}f(\pm z)\right\}$ be a vector field in the complex plane and observe
\begin{align*}
\vec{0}\in \vec{V} & = 0 + 0j\ \mbox{(zero) vector = additive identity)}.\\ 
\vec{z}\in \vec{V}f  & = a + bj,\ a,b\in \mathbb{R}.\\
\vec{-z}\in \vec{V}  & = -a - bj.\\
\vec{z} + \vec{0} &= a + bj + 0 + 0j = a + bj = \vec{z}.\\
\vec{z} + \vec{-z} &= a + bj + (-a - bj) = 0 + 0j = \vec{0}   
\end{align*}
Hence, \boxed{G(\left\{\vec{V}f(\pm z)\right\},+)} is a vector field group.
Also,   
\begin{center}
\boxed{\boldsymbol{
e^{jt}+e^{jt'} = e^{jt'}+e^{jt},
}}
\end{center}
Hence, $G$ is abelian.
\end{proof}

\vspace*{0.2cm}

\begin{figure}[!ht]
	\centering
\includegraphics[width=95mm]{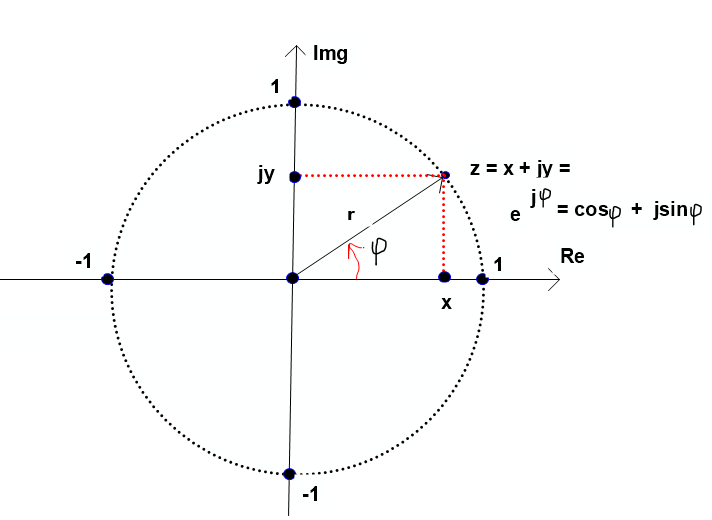}
  \caption{Polar Complex Plane}
	\label{fig:polarZ}
\end{figure}

\vspace*{0.2cm}

\begin{definition}\label{def:charFn}{\bf (Characteristic Function in Complex Plane).}\\
A characteristic mapping $\varphi:\mathbb{C}\to \mathbb{C}$ is defined by
\begin{center}
\boxed{\boldsymbol{
\varphi(z\in \mathbb{C})=\varepsilon\in [-1,1]\in 2^{\mathbb{C}}.
}}
\end{center}
Every object $A$ has a set of characteristics $\Phi(A)$, i.e.,
\begin{center}
\boxed{\boldsymbol{
\Phi(A)=\left\{\varphi(a_i):i\geq 1,a_i\in A \right\}
\in 2^{\mathbb{C}}.\ 
\mbox{\qquad \textcolor{blue}{$\blacksquare$}}
}}
\end{center}

\end{definition}

\vspace*{0.2cm}

\begin{remark}{\bf (Characteristic Functions are Holomorphic).}\\
Let $z_0,h$ be complex numbers in $\mathbb{C}$ (complex plane), for $h\neq 0$. The function $f:\mathbb{C}\to\mathbb{C}$ is holomorphic~\cite[\S 2.1, p. 8]{Stein2003}, provided 
\begin{center}
\boxed{\boldsymbol{
f'(z_0\in\mathbb{C}) = 
\mathop{lim}\limits_{h\to 0} \frac{f(z_0+h) - f(z_0)}{h}\to 0,
}}
\end{center}

\vspace*{0.5cm}

\noindent i.e., the derivative of $f$ converges to 0 when $h\to 0$. 
For example, the characteristic function $\varphi(z)$ is holomorphic, since $\varphi(z)$ is a constant and the derivative \boxed{\varphi'(z_0)=0} for every fixed complex number $z_0$ in the complex plane (see, e.g., a similar observation by S. Krantz~\cite{Krantz2008}).
\qquad \textcolor{blue}{$\blacksquare$}
\end{remark}

\vspace*{0.2cm}

\begin{definition}\label{def:dnear0}{\bf (Characteristic Distance).}\\  
Let $X,Y\in 2^{\mathbb{C}}, A\in 2^X,B\in 2^Y$ be nonempty sets in the polar complex plane and $a\in A,b\in B$ and let $\varphi(a)\in\Phi(A),\varphi(b)\in\Phi(B)$ be characteristics of objects $a\in A$ and $b\in B$.  
The characteristic distance mapping $d^{\Phi}: 2^X\times 2^Y\to\mathbb{C}$ is defined by\\

\vspace*{0.2cm}

	\begin{center}
		\boxed{\boldsymbol{
		d^{\Phi}(A,B)=\inf_{\substack{a\in A \\ b\in B}} 
		\left\{(\varphi(a) - \varphi(b))mod 2\right\} = \varepsilon\in[-1,1]\in 2^{\mathbb{C}}.
\mbox{\qquad \textcolor{blue}{$\blacksquare$}}
	}}
	\end{center}
	
\end{definition}

\vspace*{0.2cm}

\begin{remark}
The original version of the characteristic distance (limited to distances in [0,1] for characteristically near sets) was introduced in~\cite[Def. 1]{peters2025indefinitedescriptiveproximitiesinherent}.\mbox{\qquad \textcolor{blue}{$\blacksquare$}}
\end{remark}

\vspace*{0.2cm}
  
\begin{example}\label{def:char}{\bf (Stein-Weiss Characteristic~\cite[\S 5.1, p. 173]{Stein1971})).}\\
The mapping $\varphi:\mathbb{C}\to\mathbb{C}$ defined by \boxed{\varphi(t)=e^{\pm jt}} is characteristic. 
\qquad \textcolor{blue}{$\blacksquare$}
\end{example} 

\vspace*{0.2cm}

\begin{remark}
In Def.~\ref{def:charFn}, the values of the characteristic function $\varphi(z)$ are in the interval [-1,1] in the polar complex plane, which dovetails with characteristics of pairs of nonempty sets that are either characteristically near or far apart (see, e.g.,~\cite{Naimpally2013}).  That is, pairs of sets $A$ and $B$ have a characteristic distance that is in one of the following intervals:

\begin{align*}
\boldsymbol{0}\ & \mbox{if $A$ and $B$ have matching characteristics.}\\
\boldsymbol{(0,1]}\ & \ \mbox{if $A$ and $B$ are characteristically close, i.e.,}\\
   &~\mbox{if $A$ and $B$ have numerically close characteristics.}\\
\vec{z}\in \vec{V}f  & = a + bj,\ a,b\in \mathbb{R}.\\
\boldsymbol{-1} &\ \mbox{if $A,B$ have no matching characteristics, i.e.,}\\
   &\ \mbox{if $A,B$ are characteristicall far apart}\\
\boldsymbol{(-1,0)} &\ \mbox{if $A,B$ have non-matching but numerically close characteristic values.}
\end{align*}

In other words, the interval \boxed{\boldsymbol{\ [-1,1]\ }} represents the spectrum of characteristics distances between pairs of objects that are either characteristically near each other in the interval \boxed{\boldsymbol{\ [0,1]\ }} or characteristically far apart in the interval \boxed{\boldsymbol{\ [-1,0)\ }}.
\qquad \textcolor{blue}{$\blacksquare$}
\end{remark}

\vspace*{0.2cm}




	%
	
\begin{example}{\bf (Sample Characteristic Distances).}\\
Characteristic distance satisfying Def.~\ref{def:dnear0} can be found in a vector field in the polar form of the complex plane $\mathbb{C}$ (see Fig.~\ref{fig:polarZ}) for $A=\left\{a\right\},B=\left\{b\right\}\in 2^{\mathbb{C}}$, e.g.,
\begin{align*} 
d^{\Phi}(A,B)=-1 &\ \mbox{for}\ \varphi(a)=2+j, \varphi(b)=3+j.\\
d^{\Phi}(A,B)=-0.5 &\ \mbox{for}\ \varphi(a)=-2.5+j, \varphi(b)=2.0+j.\\
d^{\Phi}(A,B)=1 &\ \mbox{for}\ \varphi(a)=1+0j, \varphi(b)=0+0j.\\
d^{\Phi}(A,B)=0.5 &\ \mbox{for}\ \varphi(a)=2.5+j, \varphi(b)=2.0+j.\\
d^{\Phi}(A,B)=0 &\ \mbox{for}\ \varphi(a)=-1+j, \varphi(b)=-1+j,\varphi(a)-\varphi(b)=0.\mbox{\qquad \textcolor{blue}{$\blacksquare$}}\\
\end{align*}
\end{example}
				

\noindent In effect, $A,B$ are characteristically near,  provided\\

\vspace*{0.2cm}

\begin{center}
\boxed{\boldsymbol{ 
0\leq d^{\Phi}(A,B) \leq 1\ \mbox{\bf (Close sets)} 
}}
\end{center}

\vspace*{0.2cm}

\noindent in the first quadrant of the unit circle in the complex plane $\mathbb{C}$.  Similarly,
$A,B$ are characteristically far apart, provided\\

\vspace*{0.2cm} 

\begin{center}
\boxed{\boldsymbol{ 
1\leq d^{\Phi}(A,B) < 0\ \mbox{\bf (Far apart sets)}. 
}}
\end{center} 

\vspace*{0.2cm}

\section{Characteristic Nearness Approximation Spaces (cNAS)}
\noindent The basis for each cNAS built on a nonempty set of objects \boxed{\mathcal{O}=\left\{x_i\right\},1\leq i\leq \infty} is threefold, namely,\\ 

\vspace*{0.2cm}

\begin{compactenum}[{{\bf cNAS}-}1$^o$]
\item The identification of objects in a cNAS is guided by Axiom~\ref{axiom:cNASobject}.\\

\vspace*{0.2cm}

\begin{mdframed}[backgroundcolor=green!15]
\begin{axiom}\label{axiom:cNASobject}{\bf (Object Characteristic Axiom).}\\
The characteristic of every object is a complex number.
\qquad \textcolor{blue}{$\blacksquare$}
\end{axiom}
\end{mdframed}\mbox{}\\
 
\vspace*{0.2cm}

\item The identification of a characteristic mapping for each characteristic of an object in $x\in \mathcal{O}$, namely,\\
\vspace*{0.2cm}
 
\begin{center}
\boxed{\boldsymbol{
\varphi:\mathcal{O}\to \mathbb{C},\ \mbox{defined by}\ 
        \varphi(x)= a + bi\in \mathbb{C}, a,b\in\mathbb{R}.
}}
\end{center}

\vspace*{0.2cm}

\noindent from a set of objects to the complex plane
\item Identification of a vector of characteristics
%
  
\begin{center}
\boxed{\boldsymbol{
\vec{\Phi}(\mathcal{O}) =\left(\varphi_i(x),1\leq i\leq \infty\right)\in\mathbb{C}^n    
				\mbox{(vector of characteristics of $x\in\mathcal{O}$)}.
}}
\end{center}

\vspace*{0.2cm}

\noindent of objects in a set.  Each $\varphi_1(x)$ is a characteristic of $x\in\mathcal{O}$. \qquad \textcolor{blue}{$\blacksquare$}
\end{compactenum}\mbox{}\\

\vspace*{0.2cm}

\noindent A framework for the introduction of a cNAS is given in Def.~\ref{def:cNAS}.\\

\vspace*{0.2cm}

\begin{definition}\label{def:cNAS}{\bf (Framework for a Characteristic Nearness Approximation Space).}\\  
A Characteristic Nearness Approximation Space (cNAS) is a tuple\\ 
\begin{center} 
\boxed{\boldsymbol{
\left(\mathcal{O},\vec{\Phi}(\mathcal{O}),d^{\Phi}(A,B),x \ \widetilde{\ \Phi\ }\  y,\left[ x \right],N(\vec{\Phi}(\mathcal{O})),N_*(B),N^*(B)\right)
}}
\end{center}

\vspace*{0.2cm}

defined in the following table. 
\begin{equation*}
\begin{array}{c|l}
Symbol & Interpretation \\ \hline 
\mathcal{O} &=\left\{x_i\right\},1\leq i\leq \infty\ \mbox{(Set of objects)}.\\
A,B &\in 2^{\mathcal{O}}.\\
(x,y) &\in \mathcal{O}\times \mathcal{O}.\\
\varphi:\mathcal{O}\to \mathbb{C} &\mbox{defined by}\ 
        \varphi(x)= a + bi\in \mathbb{C}, a,b\in\mathbb{R}  
        \mbox{(characteristic mapping)}.\\
\vec{\Phi}(\mathcal{O}) &=\left(\varphi_i(x),1\leq i\leq \infty\right)\in\mathbb{C}^n\  \mbox{(vector of characteristics of $x\in\mathcal{O}$)}.\\
d^{\Phi}(A,B) &=\inf_{\substack{a\in A \\ b\in B}}\left\{ \varphi(a) - \varphi(b)\right\} = \varepsilon\in [-1,1]\in 2^{\mathbb{C}}\  \mbox{(characteristic distance)}\\ 
 A\ \widetilde{\ \Phi\ }\ B &= \left\{(a,b)\in A\times B:d^{\Phi}(A,B)=0\right\}.\ \mbox{(characteristic equivalence)}.\\
\left[x \right] &=\left\{y\in \mathcal{O}: y 
\ \widetilde{\ \Phi\ }\  x,\ 
x\in\mathcal{O} \right\}\ \mbox{(equivalence class)}.\\
N(\left[x\right]) &=\left\{[y]:y\ \widetilde{\ \Phi\ }\ x,\ 
x\in\mathcal{O}\right\} 
     \mbox{(set of neighboring partitions of $\mathcal{O}$)}.\\
N_*(B)X &=\mathop{\bigcup}\limits_{b\in B}\left\{\left[b\right]\in 2^B\right\}\subset \mathop{\bigcup}\limits_{x\in X}
          \left\{\left[x\right]\in 2^X\right\}\ \mbox{(lower approximation of $X\in 2^{\mathcal{O}}$)}.\\
N^*(B)X &=\mathop{\bigcup}\limits_{b\in B}\left\{\left[b\right]\in 2^B\right\}\cap \mathop{\bigcup}\limits_{b\in B}
          \left\{\left[x\right]\in X\right\}\ \mbox{(upper approximation of $X\in 2^{\mathcal{O}}$)}.\\ 
					\hline
\end{array}%
\end{equation*}
\end{definition}

\begin{equation*}
\text{\textit{Table} }4\text{. Characteristic Nearness Approximation Space Framework}
\end{equation*}
\begin{flushleft}

\end{flushleft}

\begin{remark}{\bf (Characteristic Foundation of a cNAS).}\\
A cNAS is built on a nonempty set of objects with measurable characteristics.  Each characteristics an object $x$ is represented by a characteristic function \boxed{\varphi(x)} with values in the complex plane. The foundation of a cNAS i is Axiom~\ref{ax:cNASbasis}.
\qquad \textcolor{blue}{$\blacksquare$}
\end{remark}\mbox{}\\

\vspace*{1.2cm}

\begin{mdframed}[backgroundcolor=green!15]

\begin{axiom}\label{ax:cNASbasis}{\bf (Characteristically Close Equivalence Class Axiom).}\\
Every nonempty set objects can be partitioned into  equivalence classes containing characteristically close objects.
\end{axiom}

\end{mdframed}

\vspace*{0.2cm}

\begin{figure}[!ht]
	\centering
\includegraphics[width=120mm]{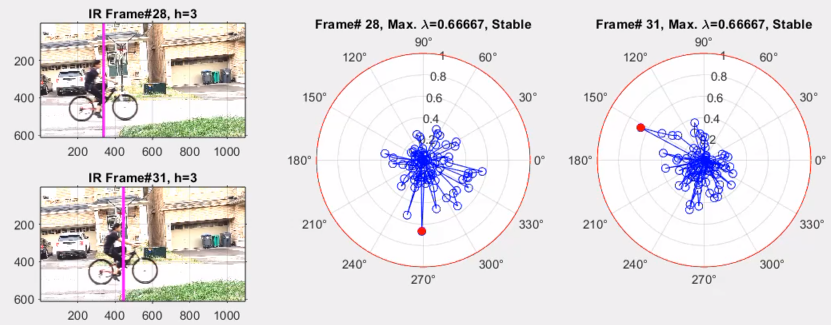}
  \caption{Characteristically near vector fields \boxed{\vec{fr}28,\vec{fr}31}}
	\label{fig:bikerVfs}
\end{figure}

\begin{table}[htbp]\label{table:sim}
\centering
\caption{Characteristics Table}
\begin{tabular}{|c|c|c|c|c|}
\hline\noalign{\smallskip}
fr  & $\varphi_{1_{\lambda_{max}}}$ & $\varphi_{2_{stable}}$ & $\varphi_{3_{unstable}}$ & 
 $frx\ \widetilde{\varphi_1,\varphi_2}\ fry$\\
\hline
\hline\noalign{\smallskip}
28 & 0.66 & 1.0 & 0.0 & $fr28\ \widetilde{\varphi_1,\varphi_2}\ fr31 \equiv d^{\Phi}(fr28,fr31)=0,$\\
& & & & $[fr28]_{\Phi}=\left\{fr28,fr31\right\}$\\
\hline\noalign{\smallskip} 
29  & 1.5 & -1.0 & -0.4 & 
$fr29\ \widetilde{\varphi_1,\varphi_2}\ fr29 
\equiv d^{\Phi}(fr29,fry)=-1,$\\
 & & & & $y=\left\{28,30,31\right\}$,\\
 & & & & $[fr29]_{\Phi}=\left\{fr29\right\}$\\
\hline\noalign{\smallskip}
30  & 0.88 & 1.0 & 0.0 & $fr30\ \widetilde{\varphi_1,\varphi_2}\ fr30 
\equiv d^{\Phi}(fr30,fry)=-1,$\\
 & & & & $y=\left\{28,29,31\right\}$,\\
 & & & & $[fr30]_{\Phi}=\left\{fr30\right\}$\\
\hline\noalign{\smallskip}
31  & 0.66 & 1.0 & -0.8 & $fr31\ \widetilde{\varphi_1,\varphi_2}\ fr30 
\equiv d^{\Phi}(fr31,fr28)=0,
$\\
 & & & & $[fr31]_{\Phi}=\left\{fr28,fr31\right\}$\\
\hline
\end{tabular}
\label{table:VfChar}
\end{table}

\vspace*{0.2cm}

\begin{example}\label{ex:sim}{\bf (Sample \boxed{frx\ \widetilde{\Phi}\ fry} Video Frame Equivalence Relations).}\\
A sample pair of characteristically near biker motion vector fields are shown in Fig.~\ref{fig:bikerVfs}.  For simplicity, we consider only three characteristics of these vector fields, namely,
\begin{align*}
\boldsymbol{\Phi} &= \boldsymbol{\left\{\varphi_{1_{\lambda_{max}}},
\varphi_{2_{stable}},
\varphi_{3_{unstable}}\right\}},i.e.,\\
\boldsymbol{\varphi_{1_{\lambda_{max}}}} &= \mbox{maximum eigenvalue in polar complex plane}.\\
\boldsymbol{\varphi_{2_{stable}}} &= 1\Leftrightarrow\lambda_{max}\in [0,1].\\
\boldsymbol{\varphi_{3_{unstable}}} &= -1\Leftrightarrow \lambda_{max}\in (-\infty,0).\\
\end{align*}

A comparison of the values for characteristic functions

\vspace*{0.2cm}

\begin{center}
\boxed{\boldsymbol{ 
\Phi = \left\{\varphi_{1_{\lambda_{max}}},
\varphi_{2_{stable}},
\varphi_{3_{unstable}}
\right\}
}}
\end{center} 

\vspace*{0.2cm}

\noindent for the vector fields in frames fr28,fr29,fr30,fr31 is given in Table~\ref{table:VfChar}.
qquad \textcolor{blue}{$\blacksquare$}
\end{example}

\vspace*{0.2cm}

\begin{example}\label{ex:classes}{\bf (Sample Video Frame Equivalence classes).}\\
From Table~\ref{table:VfChar} and Example~\ref{ex:sim}, we have the following equivalence classes:\\
\begin{center}
\begin{align*}
\boldsymbol{[fr28]} &= \left\{fr28,fr31\right\}, since \boxed{\boldsymbol{d^{\Phi}(fr28,fr31)=0}}.\\
\boldsymbol{[fr29]} &= \left\{fr29\right\}, since \boxed{\boldsymbol{d^{\Phi}(fr29,\left\{fr28,fr30,fr31\right\})=-1}}.\\
\boldsymbol{[fr30]} &= \left\{fr30\right\}, since \boxed{\boldsymbol{d^{\Phi}(fr30,\left\{fr28,fr29,fr31\right\})=-1}}.
\mbox{\qquad \textcolor{blue}{$\blacksquare$}}
\end{align*}
\end{center}
\end{example}

\vspace*{0.2cm}

\begin{example}{\bf (Sample \boxed{frx \mathop{\sim}\limits_{\Phi} fry} \& Video Frame Equivalence classes).}\\
From Table~\ref{table:VfChar} and Example~\ref{ex:classes}, we have the following nearness approximations:\\
\begin{center}
\begin{align*}
X &= \left\{[frx]: frx\in \left\{\mbox{video frames}\right\}\right\}.\\
\boldsymbol{[fr28]} &\mathop{\sim}\limits_{\Phi}  [fr31], since\ \boxed{\boldsymbol{d^{\Phi}([fr28],[fr31])=0}}.\\
\boldsymbol{[fr29]} &\mathop{\nsim}\limits_{\Phi}  [fr28], since\ \boxed{\boldsymbol{d^{\Phi}([29],[28])=-1}}.\\
\boldsymbol{[fr30]} &\mathop{\nsim}\limits_{\Phi} [fr28], since\ \boxed{\boldsymbol{d^{\Phi}([30],[28])=-1}}.
\mbox{\qquad \textcolor{blue}{$\blacksquare$}}
\end{align*}
\end{center}
\end{example}

\vspace*{0.2cm}

\begin{example}{\bf (Sample neighboring Video Frame Equivalence classes).}\\
From Table~\ref{table:VfChar} and Example~\ref{ex:classes}, we have the following neighborhoods:
\begin{center}
\begin{align*}
X &= \left\{[frx]: frx\in \left\{\mbox{video frames}\right\}\right\}.\\
\boldsymbol{N([fr28])X} &= \left\{fr28,fr31\right\}, since\ \boxed{\boldsymbol{d^{\Phi}(fr28,fr31)=0}}.\\
\boldsymbol{N([fr29])X} & \left\{fr29\right\}, since\  \boldsymbol{N([fr29]\ \nsim\ [fr28]}.\\
\boldsymbol{N([fr30])X} & \left\{fr29\right\}, since\  \boldsymbol{N([fr30]\ \nsim\ [fr28]}.\\
\mbox{\qquad \textcolor{blue}{$\blacksquare$}}
\end{align*}
\end{center}
\end{example}

\vspace*{0.2cm}

\begin{example}{\bf (Sample lower approximations of a set of Video Frame Equivalence classes).}\\
From Table~\ref{table:VfChar} and Example~\ref{ex:classes}, we have the following neighborhoods:
\begin{center}
\begin{align*}
\mathcal{O} &= \left\{[frx]: frx\in 2^{\left\{\mbox{video frames}\right\}}\right\}.\\
X &= \left\{[frx]: [28],[29],[30]\in \left\{\mbox{video frames}\right\}\right\}.\\
B_1 &= \left\{[frx]: [28]\in \left\{\mbox{video frames}\right\}\right\}.\\
B_2 &= \left\{[frx]: [29],[30]\in \left\{\mbox{video frames}\right\}\right\}.\\
\boldsymbol{N_*(\left\{B_1\right\}X} &= \mathop{\bigcup}\limits_{b\in B}\left\{\left[b\right]\in 2^{B_1}\right\}\subset \mathop{\bigcup}\limits_{x\in X}
          \left\{\left[x\right]\in 2^X\right\}\ \mbox{(lower approximation of $X$)}.\\
\boldsymbol{N^*(\left\{B_2\right\}X} &= \mathop{\bigcup}\limits_{b\in B}\left\{\left[b\right]\in 2^{B_2}\right\}\cap \mathop{\bigcup}\limits_{x\in X}
          \left\{\left[x\right]\in 2^X\right\}\ \mbox{(upper approximation of $X$)}.\mbox{
\mbox{\qquad \textcolor{blue}{$\blacksquare$}}}
\end{align*}
\end{center}
\end{example}

\begin{theorem}
Let $A,B\in 2^{\mathcal{O}}$.
$A \mathop{\sim}\limits_{\Phi}B \Leftrightarrow A\ \dnear\ B$.
\end{theorem}
\begin{proof}\mbox{}\\
$\Rightarrow$: From Table~\ref{table:VfChar}, we have
\begin{center}
\boxed{\boldsymbol{
A\mathop{\sim}\limits_{\Phi}B = \left\{(x,y)\in A\times B:d^{\Phi}(A,B)=0\right\}.
}}
\end{center}

Recall that $A$ is near $B$, provided $d^{\Phi}(A,B)=0$, i.e., the greatest lower bound of the set of differences $\abs{\varphi(a)-\varphi(b)}, a\in A,b\in B$, is zero.  In other words, for the relation between $A,B$ in $A\mathop{\sim}\limits_{\Phi}B$, the pairs of objects $(x,y)\in A\times B$ have at least on matching characteristic.\\
$\Leftarrow$: From Def.~\ref{def:dnear0} and the proof of $\Rightarrow$, we have
\begin{center}
\boxed{\boldsymbol{
A\ \dnear B \equiv \left\{(x,y)\in A\times B:d^{\Phi}(A,B)=0\right\}\ \Leftrightarrow A\mathop{\sim}\limits_{\Phi}B.
}}
\end{center}
\end{proof}

\begin{definition}\label{def:nearness}{\bf (Characteristic Nearness of Systems~\cite{peters2025}).}\\
Let $X,Y$ be a pair of systems.  For nonempty subsets $A\in 2^X, B\in 2^Y$, the characteristic nearness of $A,B$ (denoted by $A\ \rnear\ B$) is defined by\\

\vspace*{0.2cm}

\begin{center}
\boxed{\boldsymbol{
A\ \dnear\ B \Leftrightarrow
d^{\boldsymbol{\Phi}}(A,B)=0.
\mbox{\qquad \textcolor{blue}{$\blacksquare$}}  
}}
\end{center}
\end{definition}

\vspace*{0.2cm}

\begin{theorem}\label{lemma:nearLowerApprox}{\bf (Charactistically Near Lower Approximations).}\\
\boxed{N(\vec{\Phi}(A))\ \dnear\ N(\vec{\Phi}\iota(B))\ \Leftrightarrow N_*(A)X\ \dnear\ 
N_*(B)X.}  
\end{theorem}
\begin{proof}
Let $\varphi(x)\in \vec{\Phi}(X):\varphi(x) = [x]$ and let $A,B\subset X$.  Then,
from Theorem~\ref{theorem:nearness},$N(\vec{\Phi}(A))\ \dnear\ (\vec{\Phi}(B))$ implies 
\begin{center}
\boxed{\boldsymbol{
\exists [a]\in N(\vec{\Phi}(A)), [b]\in N(\vec{\Phi}(B)):N_*(A)X\ \dnear\ 
N_*(B)X.
}}
\end{center}
\end{proof}

\vspace*{1.2cm}

\begin{mdframed}[backgroundcolor=green!15]
\begin{theorem}\label{theorem:nearness}{\bf (Fundamental Theorem of Near Systems)}.\\
Let $X,Y$ be a pair of systems with $A\in 2^X,B\in 2^Y$.

\begin{center}
\boxed{\boldsymbol{
A\ \dnear\ B\ \Leftrightarrow\ \exists a\in A,b\in B:
\abs{\varphi(a)-\varphi(b)} = 0.
}}
\end{center}
\end{theorem}
\begin{proof}$\mbox{}$\\
\boxed{\Rightarrow}: From Def.~\ref{def:dnear0}, $A\ \dnear\ B$ implies that there is at least one pair $a\in A, b\in B$ such that $d^{\Phi}(A,B)=\abs{\varphi(a) - \varphi(b)}  = 0$, i.e, when $a=b$.\\
\noindent \boxed{\Leftarrow}:
Given $d^{\Phi}(A,B) =0$, we know that $\inf_{\substack{a\in A \\ b\in B}} \abs{\varphi(a) - \varphi(b)} = 0\in\mathbb{C}$. Hence, from Def.~\ref{def:nearness}, $A\ \dnear\ B$, also.  
That is, sufficient nearness of at least one pair characteristics $\varphi(a\in A), \varphi(b\in B)\in [0,1]$ such that $a=b$ indicates the characteristic nearness of the sets, i.e., we conclude $A\ \dnear\ B$. 
\end{proof}
\end{mdframed}

\vspace*{0.2cm}



\vspace*{0.2cm}

\begin{remark}
From Lemma~\ref{lemma:nearLowerApprox}, lower approximations that have characteristally near equivalence classes are characteristically near.  
\qquad \textcolor{blue}{$\blacksquare$}
\end{remark}

\begin{figure}[!ht]
	\centering
\includegraphics[width=125mm]{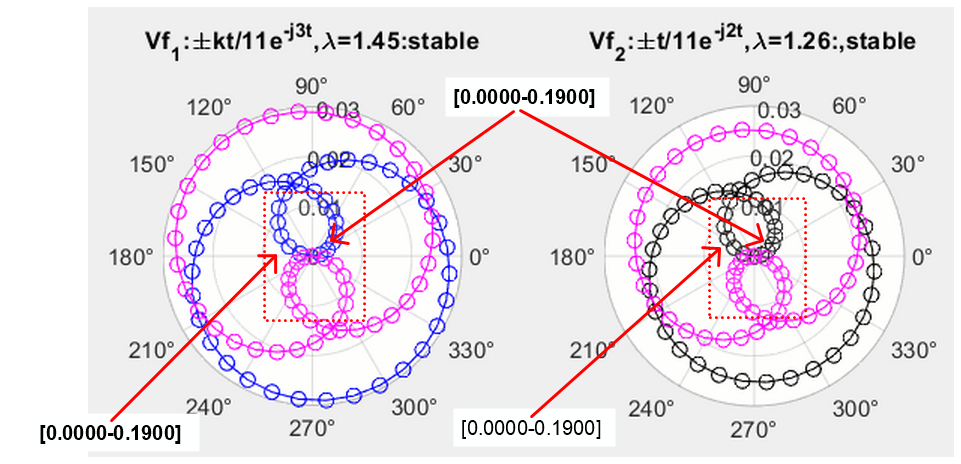}
  \caption{Sample Vector Field Equivalence Classes:
	\boxed{\vec{V}f_1,\vec{V}f_2:[0.0000-0.1900j],[0.0000+0.1900j]}}
	\label{fig:classes}
\end{figure}

\begin{table}[htbp]\label{table:eigv}
\centering
\caption{Equivalence Classes in Fig.~\ref{fig:classes}}
\begin{tabular}{|c|c|c|}
\hline
{\bf [0.0000+0.1900j]} \textbf{Qd 1}  &  {\bf [0.0000-0.1900j]} \textbf{Qd 2} & {\bf [-0.0000-0.1900j]} \textbf{Qd 3}\\
\hline
\end{tabular}
\label{table:eigv2}
\end{table}

\vspace*{0.2cm}

\section{Characteristically Near Stein-Weiss Vector Field Groups}

In this section, we introduce characteristically near vector field groups.   First, we consider a pair of groups inherent in Euler vector fields in the complex plane.

\begin{lemma}\label{lemma:symmetricFields}{\bf (Symmetric Euler Vector Field).}\\
Every Euler vector field \boxed{\vec{V}(t) = e^{jt},t\in [-k,+k]\in\mathbb{C},k\in\mathbb{R}} is symmetric about the origin in the complex plane.
\end{lemma}
\begin{proof}
Since $t\in [-k,+k]$ in $\mathbb{C}$, then $\vec{V}(t)$ is symmetric about the origin in the complex plane.
\end{proof}

\begin{lemma}\label{lemma:groupClasses}{\bf (Symmetric Vector Field Group Classes).}\\
Every Vector Field Group symmetric about the origin contains equivalence classes with real parts that differ only in terms of their sign.
\end{lemma}
\begin{proof}
From Theorem~\ref{lemma:charGroup}, a vector field in the complex plane is a group.
From Lemma~\ref{lemma:symmetricFields}, an Euler vector field
$\vec{V}(t) = e^{jt}$ defined on $t\in [-k,+k]$ in $\mathbb{C}$ is symmetric about the origin in $\mathbb{C}$. 
Hence, $\vec{V}(t)$ contains equivalence classes that differ only in terms of their sign.
\end{proof}

\begin{theorem}\label{theorem:symmetricVfgroups}{\bf (Characteristically Near Stein-Weiss Multiplicative Groups).}\\
Every pair of Stein-Weiss multiplicative groups symmetric about the origin are characteristically near groups.
\end{theorem}
\begin{proof}
From Lemma~\ref{lemma:charGroup}, 
\boxed{G_1(\left\{e^{jk_1\theta}\right\},\cdot),
       G_2(\left\{e^{jk_2\theta}\right\},\cdot)}
is a pair of Stein-Weiss multiplicative groups.  From Lemma~\ref{lemma:symmetricFields}, groups $G_1,G_2$ are symmetric about the origin.  And from Lemma~\ref{lemma:groupClasses}, groups $G_1,G_2$ contain identical classes stemming from the fact that each of the Stein-Weiss $G_1,G_2$ groups contain identical elements close to the origin in $\mathbb{C}$.  From what we have observed, we can define a pair of characteristic functions, namely,
\begin{center}
\boxed{\boldsymbol{\
\varphi(G_1) = [x],x\in\mathbb{C},\varphi(G_2) = [y],y\in\mathbb{C},\abs{\varphi(G_1)-\varphi(G_2)} = 0.
}}
\end{center}

\noindent From Theorem~\ref{theorem:nearness}, 
\boxed{\vec{\Phi}(G_1)\ \rnear\ \vec{Phi}(G_2)}.  Hence, from Lemma~\ref{lemma:nearLowerApprox}, we have a pair of neighborhoods
\boxed{N(\vec{\Phi}(G_1))\ \dnear\ N(\vec{\Phi}(G_2))\Leftrightarrow
N_*(G_1)X\ \dnear\ N_*(G_1)X}.
\end{proof}

\vspace*{0.2cm}

\begin{theorem}\label{theorem:symmetricVfgroups2}{\bf (Characteristically Near  Stein-Weiss Additive Groups).}\\
Every pair of Stein-Weiss additive groups $G_1(\left\{e^{jk_1\theta}\right\},+),
       G_2(\left\{e^{jk_2\theta}\right\},+)$ symmetric about the origin are characteristically near groups.
\end{theorem}
\begin{proof}
Symmetric with the proof of Theorem~\ref{theorem:symmetricVfgroups}.
\end{proof}

%

\section*{Acknowledgements}
Many thanks to the reviewers for their helpful comments.
We also extend our thanks to Tane Vergili for sharing here profound insights on near set theory.  In addition, we extend our thanks to Andrzej Skowron, Enze Cui, Divagar Vakeesan, Tharaka Perera, William Hankley, Brent Clark and Sheela Ramanna for sharing their insights concerning dynamical systems and to Deveci \"{O}m\"{u}r, Engin \"{O}zkan and Surabhi Tiwari (Fibonacci sequences and group theory). In some ways, this paper is a partial answer to the question 'How [temporally] Near?' put forward in 2002~\cite{PawlakPeters2002}.

This research has been supported by the Natural Sciences \&
Engineering Research Council of Canada (NSERC) discovery grant 185986 
and Instituto Nazionale di Alta Matematica (INdAM) Francesco Severi, Gruppo Nazionale per le Strutture Algebriche, Geometriche e Loro Applicazioni grant 9 920160 000362, n.prot U 2016/000036 and Scientific and Technological Research Council of Turkey (T\"{U}B\.{I}TAK) Scientific Human
Resources Development (BIDEB) under grant no: 2221-1059B211301223.
\vspace{0.2cm}

\vspace{0.2cm}

\bibliographystyle{amsplain}
\bibliography{NSrefs}
\end{document}